\pdfoutput=1
\documentclass[11pt,reqno]{amsart}
\usepackage[letterpaper,textwidth=6in,textheight=8.5in,centering]{geometry}
\usepackage{amssymb,mathtools}
\usepackage{microtype}
\usepackage[hidelinks,hypertexnames=false]{hyperref}
\theoremstyle{plain}
\newtheorem{thm}{Theorem}[section]
\newtheorem{cor}[thm]{Corollary}
\newtheorem{lem}[thm]{Lemma}
\newtheorem{prop}[thm]{Proposition}
\newtheoremstyle{cited}{6pt}{6pt}{\itshape}{}{\bfseries}{.}{5pt plus 1pt minus 1pt}{\thmname{#1} \thmnumber{#2} \thmnote{\normalfont#3}}
\theoremstyle{cited}
\newtheorem{thmC}[thm]{Theorem}
\theoremstyle{definition}
\newtheorem{defi}[thm]{Definition}
\numberwithin{equation}{section}
\hypersetup{pdftitle={Automatic constraints with few subpowers and graphoid recognition},pdfauthor={Antonios Kalampakas},pdfkeywords={constraint satisfaction, finite automata, few subpowers, Mal'tsev operations, graphoid automata}}
\date{}
\newcommand{\Mal}{Mal'tsev}
\newcommand{\Sig}{\operatorname{Sig}}
\newcommand{\Sol}{\operatorname{Sol}}
\newcommand{\AutCSP}{\operatorname{AutCSP}}
\newcommand{\pr}{\operatorname{pr}}
\newcommand{\Eq}{\operatorname{Eq}}
\newcommand{\Pclass}{\mathsf{P}}
\newcommand{\sharpP}{\mathsf{\#P}}
\newcommand{\Ftwo}{\mathbb F_2}
\newcommand{\Fp}{\mathbb F_p}
\newcommand{\aff}{\operatorname{aff}}
\newcommand{\Span}{\operatorname{span}}
\newcommand{\maj}{\operatorname{maj}}
\newcommand{\clos}[2]{\langle #1\rangle_{#2}}
\begin{document}
\title[Automatic constraints with few subpowers]{Automatic constraints with few subpowers and graphoid recognition}
\author[A. Kalampakas]{Antonios Kalampakas}
\address{American University of the Middle East, Egaila 54200, Kuwait}
\email{Antonios.Kalampakas@aum.edu.kw}
\urladdr{https://orcid.org/0000-0001-8821-6102}
\begin{abstract}
Finite automata can describe relations of unbounded arity that are exponentially larger than their descriptions. We prove that constraint satisfaction for such relations is solvable in polynomial time whenever their length slices are preserved by a common fixed edge operation on a finite domain. The algorithm computes compact representations of the complete solution relation and its projections. Its main ingredient is a polynomial-time compilation of nondeterministic finite automata into the fork witnesses and small projections required by the few-subpowers algorithm. In the Mal'tsev case, a direct proof is polynomial also when the domain and operation table are supplied as input, answering the Mal'tsev tractability question for automatic constraint satisfaction. We also characterize all invariant relations of a family of 3-edge algebras with neither Mal'tsev nor near-unanimity terms. Their normal forms combine Boolean activity constraints with affine value spaces and yield canonical quadratic-bit representations constructible from NFAs or arbitrary generators. For graphoid automata, these results give polynomial-time recognition without a graph-width restriction, effective boundary composition, and comparison of finite graph relations. The quadratic boundary bounds are optimal in the worst case. A fixed three-state example separates polynomial-time recognition from hard exact counting.
\end{abstract}
\maketitle

\section{Introduction}

A relation given by a finite automaton may contain exponentially many tuples of a specified length. This succinctness is useful for constraints on long sequences, but it prevents a direct application of algorithms whose running time is polynomial in the number of allowed tuples. The representation of a constraint must therefore be considered together with its algebraic structure.

Compact representations provide one way to make that connection. Bulatov and Dalmau proved tractability of finite-domain constraints preserved by a \Mal\ operation by maintaining small sets of tuples with prescribed fork witnesses \cite{BD2006}. Dalmau extended this approach to generalized majority-minority operations \cite{Dalmau2006}. The few-subpowers algorithm of Idziak, Markovi\'c, McKenzie, Valeriote, and Willard applies whenever the constraint language has an edge polymorphism \cite{IMMVW2010}. In these algorithms, the complete solution relation can be represented by polynomially many tuples even when it contains exponentially many solutions. Applying the algorithms to a new input representation requires an effective construction of the particular witnesses that their invariants demand.

We show that nondeterministic finite automata supply these witnesses directly. Two accepted words witness a fork at a coordinate when they agree before that coordinate and have prescribed letters there. A synchronous product of two copies of the automaton finds all possible common prefixes. Independent accepting suffixes then complete the witnesses. Bounded coordinate projections are obtained by layered reachability with prescribed letters. These two elementary constructions compile an automaton into an exact compact representation whenever its relevant length slice is closed under the given edge operation.

The resulting algorithm solves instances with arbitrarily many automaton-described constraints, unbounded arities, and repeated variables in constraint scopes. It computes the full solution relation in compact form and permits ordered projections with repeated coordinates. For a fixed domain with an $\ell$-edge operation, a relation on $k$ output coordinates is represented by $O(k^{\ell-1})$ tuples. For \Mal\ operations, the representation uses $O(k)$ tuples and $O(k^2)$ bits. The \Mal\ algorithm is also polynomial when the finite domain and the full operation table are part of the input, under the promise that the operation preserves the supplied relations.

This result addresses a question raised in the recent work of Gong, Wang, Khoussainov, and Bulatov on automatic constraint satisfaction \cite{Gong2026}. They establish tractability for several algebraic classes of automatic languages, including the Boolean affine and near-unanimity cases. Section~9 of their preprint asks whether a \Mal\ polymorphism suffices. Their model uses explicit variable scopes and relations obtained by intersecting a regular language with a specified power of a finite domain. Our \Mal\ theorem gives an affirmative answer in this model, and the edge-operation theorem extends the conclusion to all fixed finite algebras with few subpowers.

Earlier work of Chen and Grohe studies the effect of succinct relation descriptions on CSP complexity \cite{CG2006,CG2010}. Their work motivates the need for algorithms measured in the size of the description. Here the essential interface with the algebraic algorithm consists of fork witnesses and witnesses for bounded projections. These are more informative than an arbitrary small generating set. The distinction is also central to the subpower membership problem \cite{BMS2019}. Our construction starts from an automaton for the relation itself and does not require a procedure that converts arbitrary algebraic generators into a compact representation.

For a particular family of 3-edge algebras, we obtain a sharper representation by describing all invariant relations explicitly. The domain consists of a prime field together with an inactive state. Each relation is determined by a Boolean majority-closed relation on its activity bits and a compatible affine space of values. The compatibility forces a decomposition into an affine factor on always-active coordinates and independent linear factors on classes with a common variable activity. This normal form has a canonical $O(k^2)$-bit encoding, computable from either an NFA or arbitrary generators, and supports a direct calculus using 2-SAT and Gaussian elimination. The family has neither a \Mal\ term nor a near-unanimity term, and the quadratic bound is optimal.

Short definitions provide a related measure of representation size. Bul\'in and Kompatscher prove quadratic-length pp-definitions for three-element 3-edge languages \cite[Corollary~4.7]{BK2026} and ask whether short definitions can be computed efficiently from generators \cite[Question~6.3]{BK2026}. Our normal form gives an explicit construction for the displayed family over every prime field. The structural characterization and the NFA and generator conversions refine the existence bound in this case. Mixed majority and affine examples with short definitions also occur in \cite[Example~6.1]{BK2026}, following \cite[Example~2.3.2]{Brady2025}.

Our motivating application is the recognition problem for graphoid automata \cite{BK2004,BK2008}. In their canonical unitary relational interpretation, a graph imposes local constraints on its vertices, while its ordered input and output boundaries carry regular word constraints. The theorem gives polynomial-time recognition when these relations share an edge polymorphism, without a restriction on graph width. It also gives an effective calculus of boundary summaries: serial composition, parallel composition, and identification of boundary vertices can be performed directly on compact representations. The \Mal\ case and the prime-field family admit optimal quadratic-bit boundary encodings. For a fixed transition alphabet associated with the latter family, the normal form also constructs a graph with quadratically many vertices and edges defining the same boundary relation.

The general compilation theorem and the explicit normal form are the two main contributions. For general edge operations, the calculus of compact representations is established machinery, which we state with attribution. We include a direct proof of the \Mal\ case to exhibit the equality updates needed for unbounded scopes and to obtain the semiuniform result. Finally, a three-element example shows that the recognition theorem does not imply tractable exact counting. Its hardness follows from the weighted Boolean counting dichotomy of Dyer, Goldberg, and Jerrum \cite{DGJ2009}.

\section{Regular constraints and witness extraction}\label{sec:regular}

Let $D$ be a nonempty finite set and put $d=|D|$. For $r\geq 0$, write $[r]=\{1,\ldots,r\}$ and identify a word of length $r$ over $D$ with an element of $D^r$. For $u\in D^r$, the notation $u_{<i}$ denotes its prefix of length $i-1$. If $J\subseteq[r]$, then $\pr_J(u)$ lists the coordinates in increasing order. We write $\pr_J(R)=\{\pr_J(u):u\in R\}$ for the corresponding projection of a relation. Ordered coordinate lists, including repeated coordinates, will be treated explicitly when constructing output relations.

An operation $f:D^t\to D$ acts coordinatewise on tuples of a common length. It \emph{preserves} $R\subseteq D^r$ if $f(u^1,\ldots,u^t)\in R$ whenever $u^1,\ldots,u^t\in R$. We also say that $R$ is $f$-closed or that $f$ is a polymorphism of $R$. Intersections, Cartesian products, and projections of $f$-closed relations are $f$-closed. Equality is preserved by every operation. If $f$ is idempotent, every singleton relation is also preserved by $f$.

An epsilon-free NFA is a tuple $M=(Q,Q_0,\Delta,Q_f)$, where $Q_0,Q_f\subseteq Q$ and $\Delta\subseteq Q\times D\times Q$. Multiple initial states and multiple accepting paths are allowed. For an explicit arity $r$, set
\[
 R(M,r)=L(M)\cap D^r.
\]
Allowing epsilon transitions does not change our results, since they can be eliminated in polynomial time before the constructions below.

An \emph{automatic constraint instance} consists of a finite variable set $V$ and constraints $(\mathbf v_j,M_j)$, where $\mathbf v_j\in V^{r_j}$ is an explicit ordered scope. A map $h:V\to D$ is a solution if $h(\mathbf v_j)\in R(M_j,r_j)$ for every $j$. Scopes may repeat variables, and different constraints may use different NFAs. We use the size parameter
\begin{equation}\label{eq:size}
 N=|V|+\sum_j\bigl(1+r_j+|Q_j|+|\Delta_j|\bigr).
\end{equation}
All polynomial-time assertions are with respect to ordinary explicit encodings of these data. The arity is given by the length of the scope. It is not a binary-encoded instruction to introduce exponentially many positions.

For a fixed ordering of $V$, let $\Sol(\mathcal I)\subseteq D^{|V|}$ denote the solution relation of an instance $\mathcal I$. For an ordered list $\mathbf b=(b_1,\ldots,b_k)\in V^k$, define
\begin{equation}\label{eq:boundarysol}
 \Sol_{\mathbf b}(\mathcal I)
 =\{(h(b_1),\ldots,h(b_k)):h\in\Sol(\mathcal I)\}.
\end{equation}
The list $\mathbf b$ may be empty and may contain repetitions. There are two nullary relations, $\varnothing$ and $\{()\}$. We store them as distinct flags throughout.

\begin{defi}\label{def:fork}
For $R\subseteq D^r$ with $r\geq1$, a \emph{fork} $(i,a,b)\in[r]\times D^2$ is witnessed by $u,v\in R$ if
\[
 u_{<i}=v_{<i},\qquad u_i=a,\qquad v_i=b.
\]
The set of all forks of $R$ is denoted by $\Sig(R)$. Diagonal forks $(i,a,a)$ are included.
\end{defi}

The following lemma supplies the first part of the compilation. It does not require preservation by any operation.

\begin{lem}[Fork extraction]\label{lem:forkcompiler}
Given an NFA $M$ over $D$ and $r\geq1$, one can compute in polynomial time a set $F\subseteq R(M,r)$ satisfying
\[
 \Sig(F)=\Sig(R(M,r)),\qquad |F|\leq 2rd^2.
\]
Witness pairs for every realized fork can be returned with $F$.
\end{lem}
\begin{proof}
For each $t\in\{0,\ldots,r\}$, compute the set $P_t\subseteq Q^2$ of pairs reachable from $Q_0^2$ by a common word of length $t$. Thus $P_0=Q_0^2$, and $(u,u')\in P_{t+1}$ exactly when some $(q,q')\in P_t$ and $a\in D$ satisfy
\[
 (q,a,u)\in\Delta,\qquad(q',a,u')\in\Delta.
\]
Store a common prefix witnessing each reachable pair. This is reachability in a layered synchronous square of $M$.

Independently, for every $q\in Q$ and $t\leq r$, compute whether an accepting path of length $t$ starts at $q$. Store one corresponding suffix when it exists. The base case is membership in $Q_f$, and the recurrence follows the transitions backwards.

A fork $(i,a,b)$ is realized precisely when there are states $q,q',u,u'$ such that
\begin{equation}\label{eq:forktest}
 \begin{gathered}
 (q,q')\in P_{i-1},\qquad
 (q,a,u),(q',b,u')\in\Delta,\\
 u\text{ and }u'\text{ each admit an accepting suffix of length }r-i.
 \end{gathered}
\end{equation}
If these conditions hold, concatenate the stored common prefix, the respective letters, and the two stored suffixes. The resulting words are accepted and witness the fork. Conversely, accepting paths for any witnessing pair yield the states in~\eqref{eq:forktest}.

Keep one pair for each realized fork. There are at most $rd^2$ such forks, so at most $2rd^2$ words are retained. Since all retained words belong to $R(M,r)$, no additional fork outside its signature is introduced.

For completeness, let $s=|Q|$ and $a=|\Delta|$. There are $O(rs^2)$ layered state pairs. One can scan pairs of transitions at each layer and retain explicit words for the witnesses. A conservative bound for the resulting construction is $O(r^2(s+a+d)^2)$ elementary operations. Only polynomial space is required. NFA ambiguity affects the choices of paths, but the output consists of words, so repeated accepting paths do not create additional relation elements.
\end{proof}

\begin{lem}[Projection witnesses]\label{lem:pincompiler}
Given $M$, $r\geq0$, a set $J\subseteq[r]$, and $a\in D^J$, one can decide in polynomial time whether $a\in\pr_J(R(M,r))$ and return an accepted word witnessing membership. For each fixed integer $c$, witnesses for all realized projections onto at most $c$ coordinates can be computed in polynomial time.
\end{lem}
\begin{proof}
Use the layered graph with vertices $(i,q)$ for $0\leq i\leq r$ and $q\in Q$. At a prescribed position $i\in J$, retain only transitions with label $a_i$. At all other positions, retain every transition. Reachability from layer-zero initial states to layer-$r$ final states decides the question and supplies a witnessing word. The number of tests for all projections onto at most $c$ coordinates is
\[
 \sum_{t=0}^{\min(c,r)}\binom rt d^t,
\]
which is polynomial for fixed $c$. Each individual test takes $O(r(|Q|+|\Delta|))$ time with path reconstruction.
\end{proof}

The two constructions expose different information. Projection witnesses describe what can occur on a small set of positions. Fork witnesses additionally record which values can be exchanged after a common prefix. It is this second piece of information that permits compact generation for \Mal\ operations and, together with the bounded projections, for edge operations.

\section{The Mal'tsev case}\label{sec:maltsev}

A \Mal\ operation on $D$ is a ternary operation $p$ satisfying
\begin{equation}\label{eq:maltsev}
 p(x,y,y)=x,\qquad p(y,y,x)=x.
\end{equation}
In particular, $p$ is idempotent. We write $\clos{F}{p}$ for the closure of a set of tuples $F$ under coordinatewise applications of $p$.

For a $p$-closed relation $R\subseteq D^r$ of positive arity, a \emph{frame} is a set $F\subseteq R$ with $\Sig(F)=\Sig(R)$. A frame with at most $2rd^2$ tuples is a compact representation. Any finite frame can be pruned to this size by choosing a witness pair for each of its forks. The empty relation is represented by the empty frame. At arity zero we use the flags specified in Section~\ref{sec:regular}.

The generation property of frames is classical \cite{BD2006}. We recall its proof because it explains why Lemma~\ref{lem:forkcompiler} gives an exact compilation.

\begin{lem}\label{lem:maltsevgen}
If $R$ is $p$-closed and $F\subseteq R$ is a frame, then $\clos{F}{p}=R$.
\end{lem}
\begin{proof}
Only the inclusion $R\subseteq\clos{F}{p}$ requires proof. Suppose $R$ is nonempty and has positive arity. The diagonal forks imply that $F$ is nonempty. Fix a target $u\in R$ and start with any $t\in F$. Inductively suppose $t_{<i}=u_{<i}$. Since $(i,t_i,u_i)\in\Sig(R)$, there are $a,b\in F$ with
\[
 a_{<i}=b_{<i},\qquad a_i=t_i,\qquad b_i=u_i.
\]
The tuple $p(t,a,b)$ agrees with $t$ before $i$ by the first identity in~\eqref{eq:maltsev}, and its $i$th coordinate is $u_i$ by the second. After $r$ steps the target has been constructed. Empty and nullary relations are immediate.
\end{proof}

\begin{cor}\label{cor:maltsevcompile}
If $R(M,r)$ is preserved by $p$, Lemma~\ref{lem:forkcompiler} computes a compact representation of it in polynomial time. No call to the operation table is required during the compilation itself.
\end{cor}

We next give the algebraic operations used after compilation. The calculus of compact representations is part of the established \Mal\ CSP machinery. The role of compact intersection is discussed in \cite[Section~1]{BMS2019}, and a direct presentation of the product, identification, and projection calculus appears in \cite[Section~1.8]{Brady2025}. We provide the required special cases with their polynomial dependence on the domain size.

\begin{lem}[Small projections]\label{lem:smallproj}
Let $F$ frame $R\subseteq D^n$. For a fixed number $c$ of selected coordinates, their projection of $R$ and one full tuple witness for each projected tuple can be computed in time polynomial in $n$, $|F|$, and $d$, given the table of $p$.
\end{lem}
\begin{proof}
Begin with the projected tuples from $F$ and close this set under $p$. Retain a full tuple witness for each projected tuple. A new projected tuple is witnessed by applying $p$ to the full witnesses of the three tuples that produced it. All witnesses therefore remain in $R$. Projection commutes with coordinatewise generation, so Lemma~\ref{lem:maltsevgen} proves that the resulting set is the desired projection. There are at most $d^c$ projected tuples. Closing this bounded power under a table-given ternary operation is polynomial for fixed $c$.
\end{proof}

\begin{lem}[Prefix restriction]\label{lem:prefix}
Given a frame for $R\subseteq D^n$ and $a\in D^j$, one can compute a frame for
\[
 R_a=\{t\in R:t_1=a_1,\ldots,t_j=a_j\}
\]
in time polynomial in $n$, $d$, and the input frame size.
\end{lem}
\begin{proof}
Fix the coordinates successively. Suppose a frame $F$ for a relation $S$ already fixes positions before $j$, and position $j$ is to be fixed to $v$. If $S$ has no tuple with value $v$ there, the answer is empty. This is detected by inspecting $F$, since the diagonal forks ensure that $\pr_{\{j\}}(F)=\pr_{\{j\}}(S)$.

Otherwise retain an anchor $z\in F$ with $z_j=v$. For each $(i,a,b)\in\Sig(S)$ with $i>j$, choose witnesses $u,w\in F$. By Lemma~\ref{lem:smallproj}, find a tuple $t\in S$ with $t_j=v$ and $t_i=a$, if one exists. In that case retain
\begin{equation}\label{eq:prefixpair}
 t,\qquad p(t,u,w).
\end{equation}
They agree before $i$. At position $j$, the two original witnesses agree because $j<i$, so the second tuple also has value $v$. At position $i$, its value is $p(a,a,b)=b$. Thus~\eqref{eq:prefixpair} witnesses the required fork in the restricted relation.

Conversely, a fork of the restricted relation at a position $i>j$ is a fork of $S$, and its first witness supplies a tuple with values $v$ at $j$ and $a$ at $i$. The search therefore succeeds for every such fork. Positions at most $j$ have only their forced diagonal forks, all witnessed by the anchor. Pruning gives the next frame. Iterate this step. Only binary projections are used, and the number of steps and forks is polynomial.
\end{proof}

\begin{lem}[Equality update]\label{lem:equality}
Given a frame for $R\subseteq D^n$ and coordinates $\alpha,\beta\in[n]$, one can compute a frame for
\[
 R'=\{t\in R:t_\alpha=t_\beta\}
\]
in time polynomial in $n$, $d$, and the input frame size.
\end{lem}
\begin{proof}
For every $(i,a,b)\in\Sig(R)$, use a projection on the at most three coordinates $\{\alpha,\beta,i\}$ to search for $t\in R'$ with $t_i=a$. If there is no such tuple, this fork cannot survive. If one is found, apply Lemma~\ref{lem:prefix} to fix $t_{<i}$ in $R$. Search the same small projection of this restricted relation for a tuple $u$ satisfying
\[
 u_\alpha=u_\beta,\qquad u_i=b.
\]
Retain $t,u$ when the second search succeeds.

Every retained pair lies in $R'$ and witnesses the indicated fork. To see that no fork of $R'$ is missed, let $v,w\in R'$ witness $(i,a,b)$. For any first tuple $t$ found by the algorithm, the tuple $p(t,v,w)$ lies in $R'$ because equality is preserved by $p$. Before position $i$ it agrees with $t$, and at position $i$ it equals $p(a,a,b)=b$. Hence the second search succeeds regardless of the choice of $t$. The retained pairs have exactly the signature of $R'$. Prune them to obtain a frame. There are at most $nd^2$ searches, and the preceding lemmas give polynomial running time.
\end{proof}

\begin{lem}[Products and projections]\label{lem:maltsevcalculus}
Cartesian product and arbitrary ordered projection, including coordinate repetitions, can be performed on frames in polynomial time.
\end{lem}
\begin{proof}
Let $F,G$ frame nonempty relations $R\subseteq D^n$ and $S\subseteq D^m$. Choose anchors $a\in F$ and $b\in G$. In the concatenated coordinate order,
\begin{equation}\label{eq:anchored}
 (F\times\{b\})\ \cup\ (\{a\}\times G)
\end{equation}
frames $R\times S$. Forks in the first block come from $R$, while forks in the second block come from $S$ with the first block held fixed. Empty and nullary factors are handled directly.

Projection of a frame to a prefix of its coordinates gives a frame for that projected relation. Indeed, any fork in the projected relation lifts to two original tuples that agree on the same earlier coordinates. Its witnesses in the original frame project to witnesses of the fork.

For a general ordered list $(i_1,\ldots,i_k)$, adjoin $k$ new coordinates before the old coordinates, initially ranging over $D^k$. Impose $z_j=x_{i_j}$ for each $j$ using Lemma~\ref{lem:equality}, and project to the first $k$ coordinates. This also permits repetitions and permutations. A frame for $D^k$ consists of a fixed base tuple and the tuples obtained by changing at most one coordinate. Thus all required initial representations are polynomial in size.
\end{proof}

\begin{thm}[Automatic \Mal\ constraints]\label{thm:maltsev}
Suppose a common \Mal\ operation $p$ preserves all slices $R(M_j,r_j)$ in an automatic constraint instance $\mathcal I$. One can decide satisfiability and compute a frame for $\Sol(\mathcal I)$ in polynomial time. For any explicit ordered boundary list $\mathbf b$, one can also compute a frame for $\Sol_{\mathbf b}(\mathcal I)$ in polynomial time. The theorem holds both for fixed $D,p$ and when $D$ and the full table of $p$ are supplied as input, under the preservation promise.
\end{thm}
\begin{proof}
Compile each constraint by Corollary~\ref{cor:maltsevcompile}. Start with a frame for $D^n$, where $n=|V|$. Maintain a frame for the solution relation of the constraints already processed.

Suppose the next constraint has relation $S\subseteq D^r$ and scope $(v_{i_1},\ldots,v_{i_r})$. Form a frame for the product of the current relation with $S$, using the coordinate order
\[
 (x_1,\ldots,x_n,y_1,\ldots,y_r).
\]
For each $j\in[r]$, impose $x_{i_j}=y_j$. Project to the first $n$ coordinates. The resulting relation consists exactly of the assignments satisfying the previous constraints and the new one. Each update uses only a binary equality, even when $r$ is unbounded. Repeated variables in the scope simply cause several $y$-coordinates to be identified with one $x$-coordinate.

There are $\sum_j r_j$ equality updates, and the intermediate arity is at most $n+\max_j r_j$. The preceding lemmas are polynomial in this arity and in $d$. The final frame is empty exactly when the instance is unsatisfiable, and every tuple it contains is a solution. The ordered boundary relation is computed by Lemma~\ref{lem:maltsevcalculus}.

For the semiuniform assertion, the Mal'tsev identities can be checked directly from the table. All uses of closure take place in powers of dimension at most three, and their size is bounded by $d^3$. Hence the entire algorithm is polynomial in the instance encoding, the requested boundary length, and the operation-table size. Preservation of the NFA-described slices is the stated promise.
\end{proof}

\begin{cor}\label{cor:automaticquestion}
Let $M$ be a fixed NFA over a finite domain. If a \Mal\ operation preserves every slice $L(M)\cap D^r$, then $\AutCSP(M)$ is in $\Pclass$.
\end{cor}
\begin{proof}
Use $M_j=M$ in Theorem~\ref{thm:maltsev}. This is the automatic constraint model of \cite[Definitions~2.5--2.7]{Gong2026}, so the result answers its \Mal\ tractability question.
\end{proof}

The output is a representation by actual solutions. It is not a list of all solutions, and the algorithm never closes a frame inside the full power $D^n$. The bounded projections in Lemma~\ref{lem:smallproj} are the only explicitly generated subpowers.

\section{Edge operations and few subpowers}\label{sec:edge}

Fix an integer $\ell\geq2$. An \emph{$\ell$-edge operation} is an operation $e:D^{\ell+1}\to D$ satisfying
\begin{equation}\label{eq:edge}
 \begin{aligned}
 e(x,x,y,y,\ldots,y)&=y,\\
 e(x,y,x,y,\ldots,y)&=y,\\
 e(y,\ldots,y,\underset{j}{x},y,\ldots,y)&=y
       &&(4\leq j\leq\ell+1).
 \end{aligned}
\end{equation}
The last family is absent for $\ell=2$. These operations are idempotent. If $p$ is \Mal, then $e(x,y,z)=p(y,x,z)$ is a $2$-edge operation. A near-unanimity operation of arity $\ell$ gives an $\ell$-edge operation by ignoring the first of $\ell+1$ arguments.

A finite algebra has \emph{few subpowers} if the number of subalgebras of its $r$th power is bounded by $2^{P(r)}$ for some polynomial $P$. The existence of an edge term characterizes this property \cite[Theorem~3.4]{IMMVW2010}. We use its constructive representation theorem. Throughout this section, $D$ and $e$ are fixed, so the polynomial exponent may depend on $\ell$.

To state the representation invariant, fix derived term operations $\mu(x,y)$, $p_e(x,y,z)$, and $s(x_1,\ldots,x_\ell)$ supplied by \cite[Lemma~3.5]{IMMVW2010}. They satisfy
\begin{equation}\label{eq:derived}
 \begin{gathered}
 p_e(x,y,y)=x,\qquad p_e(x,x,y)=\mu(x,y),\qquad
 \mu(x,\mu(x,y))=\mu(x,y),\\
 s(y,x,\ldots,x)=\mu(x,y),\qquad
 s(x,\ldots,x,\underset{j}{y},x,\ldots,x)=x\quad(2\leq j\leq\ell).
 \end{gathered}
\end{equation}
These are fixed compositions of $e$. A pair $(a,b)$ is a \emph{minority pair} if $\mu(a,b)=b$. Let $\Sig_\mu(R)$ consist of the forks of $R$ whose last two entries form a minority pair.

\begin{defi}\label{def:edgeframe}
For an $e$-closed relation $R\subseteq D^r$ of positive arity, an \emph{edge representation} is a set $F\subseteq R$ such that
\[
 \Sig_\mu(F)=\Sig_\mu(R),\qquad
 \pr_J(F)=\pr_J(R)\quad\text{for every }J\subseteq[r]\text{ with }|J|<\ell.
\]
We call the representation compact when its number of tuples is at most
\begin{equation}\label{eq:edge-bound}
 B_{\ell,d}(r)=2rd^2+\sum_{t=0}^{\min(\ell-1,r)}\binom rt d^t.
\end{equation}
Empty and nullary relations again have explicit representations.
\end{defi}

Our size convention retains witnesses for every projection of size less than $\ell$. The standard convention may retain just those of size $\min(\ell-1,r)$, since they also cover the smaller projections. Either convention gives $O(r^{\ell-1})$ tuples for fixed $D,e$, and a representation in our convention can be pruned to the standard convention in polynomial time.

We use two results from the few-subpowers algorithm in the following form.

\begin{thmC}[{\cite[Theorem~3.10 and Section~4]{IMMVW2010}}]\label{thm:classicaledge}
An edge representation $F$ of an $e$-closed relation $R$ satisfies $\clos{F}{e}=R$. Given a compact representation of $R\subseteq D^n$, distinct selected coordinates $i_1,\ldots,i_t$, and an explicitly listed $e$-closed relation $S\subseteq D^t$, the algorithm \textnormal{Next} computes a compact representation of
\[
 \{u\in R:(u_{i_1},\ldots,u_{i_t})\in S\}
\]
in polynomial time for fixed $D,e$.
\end{thmC}

In our applications of \textnormal{Next}, $S$ is either the binary equality relation or a unary singleton. Its explicit size is therefore bounded by $d$. No large NFA-described relation is passed as the table argument of this algorithm.

\begin{lem}[Edge compilation]\label{lem:edgecompiler}
Given $M$ and $r$, if $R(M,r)$ is preserved by $e$, a compact edge representation of it can be computed in polynomial time.
\end{lem}
\begin{proof}
Apply Lemma~\ref{lem:forkcompiler} and retain witness pairs for the forks whose letter pairs satisfy $\mu(a,b)=b$. Apply Lemma~\ref{lem:pincompiler} for every $J\subseteq[r]$ with $|J|<\ell$, retaining one accepted word for every realized assignment on $J$. The union lies in $R(M,r)$ and has all the minority forks and all the required projections. Its size is bounded by~\eqref{eq:edge-bound}. By Theorem~\ref{thm:classicaledge}, it generates the entire slice. The number of reachability tests is polynomial because $D$ and $\ell$ are fixed. Nullary slices are decided by acceptance of the empty word.
\end{proof}

For use in the main algorithm, we spell out how products and projections preserve the stronger edge representation invariant.

\begin{lem}[Edge products and projections]\label{lem:edgecalculus}
Products, coordinate identifications, and arbitrary ordered projections of $e$-closed relations can be computed from their compact edge representations in polynomial time.
\end{lem}
\begin{proof}
Consider nonempty relations $R\subseteq D^n$ and $S\subseteq D^m$ with representations $F,G$. Form the full product $F\times G$ of the two compact sets. Its size is polynomial. Minority forks of $R\times S$ in the first coordinate block are supplied by witness pairs from $F$ with any member of $G$ fixed. Those in the second block are supplied by witness pairs from $G$ with any member of $F$ fixed. Thus $F\times G$ has the required minority signature.

For its projection part, fix a set $J$ of fewer than $\ell$ coordinates of the product. Split it into $J_R$ and $J_S$ in the two factors. The desired projection is
\[
 \pr_J(R\times S)=\pr_{J_R}(R)\times\pr_{J_S}(S).
\]
Both factor projections are already visible in $F$ and $G$, so every pair of realized projected tuples has a witness in $F\times G$. Hence this product also has all the required small projections. Prune it by choosing a witness pair for each minority fork and one witness for each realized small projection. The result is a compact representation of $R\times S$.

Coordinate identification is an application of \textnormal{Next} with $S=\Eq_D=\{(a,a):a\in D\}$. For prefix projection, all required small projections are inherited. A minority fork in the projected relation lifts to a fork in the original relation at the same position and with the same values, so it too has projected witnesses. Therefore a prefix projection of an edge representation remains an edge representation.

For an arbitrary ordered projection, adjoin fresh coordinates before the original ones, identify them with the selected coordinates, and take a prefix projection. A compact representation of $D^k$ is obtained from a base tuple by allowing changes on fewer than $\ell$ coordinates. This set supplies every small projection and every minority fork. All construction sizes are polynomial. Empty and nullary relations are handled separately.
\end{proof}

\begin{thm}[Automatic constraints with an edge polymorphism]\label{thm:edge}
Fix a finite domain $D$ and an $\ell$-edge operation $e$ on $D$. Automatic constraint satisfaction is solvable in polynomial time under the promise that $e$ preserves each used length slice. A compact edge representation of the complete solution relation can be computed in polynomial time. For any explicit ordered boundary list of length $k\geq1$, a representation of its solution relation with at most $B_{\ell,d}(k)$ tuples can also be computed in polynomial time. Repeated variables are allowed in both constraint scopes and boundary lists.
\end{thm}
\begin{proof}
Compile all slices by Lemma~\ref{lem:edgecompiler}. Use the constraint-insertion procedure from the proof of Theorem~\ref{thm:maltsev}, replacing its frame operations by Lemma~\ref{lem:edgecalculus}. All equality relations used by \textnormal{Next} have constant size. There are only $\sum_j r_j$ such updates, and intermediate arities remain at most $|V|+\max_j r_j$. A final ordered projection computes the requested boundary relation. Pruning gives the bound~\eqref{eq:edge-bound}.
\end{proof}

The theorem also applies when all slices are invariant under the basic operations of a fixed finite algebra with few subpowers. Every term operation of that algebra preserves the slices, and an edge term exists by the characterization quoted above. This is a fixed-algebra statement. It does not assert a uniform polynomial-time procedure for discovering edge terms from arbitrary input algebras.

\begin{cor}[Comparison and enumeration]\label{cor:comparison}
For two automatic instances over the same fixed $D,e$ and ordered boundary lists of the same length, inclusion and equality of their boundary solution relations are decidable in polynomial time. A failed inclusion has a computable witnessing boundary tuple. Solutions of an automatic instance can be enumerated with polynomial delay and polynomial space. The same conclusions hold semiuniformly in the \Mal\ setting of Theorem~\ref{thm:maltsev}.
\end{cor}
\begin{proof}
Let $F$ represent $R$, and let $S$ be $e$-closed. Since $\clos{F}{e}=R$,
\[
 R\subseteq S\quad\Longleftrightarrow\quad F\subseteq S.
\]
Membership of a tuple in $S$ is decided by successive unary restrictions of its representation using \textnormal{Next}. Thus a tuple of $F$ outside $S$ witnesses a failed inclusion. Equality is mutual inclusion. Nullary cases are tested from their flags.

For enumeration, fix a variable order and traverse the assignment tree depth first. Test a prefix by successive singleton restrictions, and discard it precisely when the resulting relation is empty. Every surviving prefix extends to a solution. Between consecutive solutions, the traversal performs at most polynomially many prefix tests, including at most $d$ trials at each of $|V|$ levels. Storing a polynomial-size representation at each level uses polynomial space. In the \Mal\ setting, Lemma~\ref{lem:prefix} implements the restrictions with polynomial dependence on $d$.
\end{proof}

\begin{prop}[Checking preservation for DFAs]\label{prop:dfa}
Fix the arity $t$ of an operation $f:D^t\to D$. Given its table and a complete DFA with $s$ states, preservation of every length slice by $f$ can be decided in polynomial time. Preservation at a single explicitly specified length is also decidable in polynomial time.
\end{prop}
\begin{proof}
Use $t+1$ synchronous copies of the DFA. At each step choose letters $a_1,\ldots,a_t\in D$. The first $t$ copies read these letters, and the last copy reads $f(a_1,\ldots,a_t)$. A reachable product state with the first $t$ coordinates accepting and the last rejecting witnesses a failure of preservation. Conversely, any failure supplies such a path. The product has $s^{t+1}$ states and at most $d^t$ outgoing choices per state. Restricting reachability to the specified number of layers tests a single length.
\end{proof}

For general NFAs our algorithms use the preservation promise. They construct neither a complement automaton nor a determinization. This is why Proposition~\ref{prop:dfa} is stated separately rather than included in the input processing of Theorems~\ref{thm:maltsev} and~\ref{thm:edge}.

\section{Canonical quadratic representations for a 3-edge family}\label{sec:normalforms}

The general bound~\eqref{eq:edge-bound} gives $O(k^3)$ bits for a 3-edge operation. We now obtain a canonical quadratic representation for an explicit family by characterizing its invariant relations. Fix a prime $p$ and let $D_p=\{\bot\}\mathbin{\dot\cup}\Fp$, with an inactive state $\bot$ distinct from the active field element $0$. Define activity and value maps by
\[
 \chi(\bot)=0,\quad \chi(a)=1\ (a\in\Fp),\qquad
 \eta(\bot)=0,\quad \eta(a)=a\ (a\in\Fp).
\]
Both maps act coordinatewise. Activity bits $q\in\{0,1\}^k$ and values $v\in\Fp^k$ describe a tuple over $D_p$ exactly when $v_i=0$ at every position with $q_i=0$. Let $D_q$ be the diagonal linear map $v\mapsto(q_iv_i)_{i=1}^k$, and define
\begin{equation}\label{eq:activeedge}
 e_p(a,b,c,d)=
 \begin{cases}
  \bot,&\maj(\chi(b),\chi(c),\chi(d))=0,\\
  -\eta(a)+\eta(b)+\eta(c),&\maj(\chi(b),\chi(c),\chi(d))=1.
 \end{cases}
\end{equation}
The second case is interpreted in the active copy of $\Fp$. The prime and this operation remain fixed throughout the section.

\begin{lem}\label{lem:activemask}
The operation $e_p$ is a 3-edge operation, and the term $\mu_p(x,y)=e_p(x,x,y,x)$ satisfies
\begin{equation}\label{eq:activemask}
 \chi(\mu_p(x,y))=\chi(x),\qquad
 \eta(\mu_p(x,y))=\chi(x)\eta(y).
\end{equation}
The algebra $(D_p,e_p)$ has neither a \Mal\ term nor a near-unanimity term.
\end{lem}
\begin{proof}
In each of $e_p(x,x,y,y)$, $e_p(x,y,x,y)$, and $e_p(y,y,y,x)$, the output activity is $\chi(y)$ and the unmasked field value is $\eta(y)$. These are the identities~\eqref{eq:edge} for $\ell=3$. In $\mu_p(x,y)$ the activity majority is $\chi(x)$ and the unmasked value is $\eta(y)$, proving~\eqref{eq:activemask}.

The activity quotient has only monotone term operations, since its basic operation is the Boolean majority of the last three inputs. A Boolean \Mal\ operation would take values $1$ on $(0,0,1)$ and $0$ on $(0,1,1)$, contrary to monotonicity. Thus no \Mal\ term exists. On the active subalgebra $\Fp$, every term is a linear combination $\sum_i\lambda_i x_i$ with $\sum_i\lambda_i=1$. The near-unanimity identities, applied to one input $1$ and all other inputs $0$, would force every $\lambda_i=0$. Hence no near-unanimity term exists either.
\end{proof}

For a nonempty $e_p$-closed relation $R\subseteq D_p^k$, put
\begin{equation}\label{eq:activecanonical}
 Q_R=\chi(R),\qquad W_R=\aff_{\Fp}(\eta(R)).
\end{equation}
The use of a prime field ensures that additive subgroups of its vector spaces are linear subspaces.

\begin{lem}[Activity fibres]\label{lem:activefibres}
The relation $Q_R$ is majority-closed. For every $q\in Q_R$, its value fibre satisfies
\begin{equation}\label{eq:activefibre}
 \{\eta(t):t\in R,\ \chi(t)=q\}=D_qW_R\subseteq W_R.
\end{equation}
Consequently,
\begin{equation}\label{eq:activenormal}
 R=\{t\in D_p^k:\chi(t)\in Q_R,\ \eta(t)\in W_R\}.
\end{equation}
\end{lem}
\begin{proof}
The activity map is a homomorphism, so $Q_R$ is majority-closed. Write $T_q$ for the fibre on the left of~\eqref{eq:activefibre}. Applying $e_p$ within one activity fibre shows that $T_q$ is closed under $(u,v,w)\mapsto-u+v+w$. Translating by any of its elements gives a set containing zero and closed under addition and additive inverses. Hence $T_q$ is a nonempty affine space over $\Fp$.

The masking term~\eqref{eq:activemask} gives $D_qT_s\subseteq T_q$ for all $q,s\in Q_R$. Since $D_q$ is linear and $T_q$ is affine,
\[
 D_qW_R=\aff\left(\bigcup_{s\in Q_R}D_qT_s\right)\subseteq T_q.
\]
Every $v\in T_q$ belongs to $W_R$ and has $D_qv=v$, giving equality. This also proves $D_qW_R\subseteq W_R$. If a tuple $t\in D_p^k$ satisfies the two conditions in~\eqref{eq:activenormal}, then its value vector is fixed by $D_{\chi(t)}$ and therefore belongs to the appropriate fibre. This proves the reverse inclusion in~\eqref{eq:activenormal}. The forward inclusion follows from the definitions.
\end{proof}

Conversely, suppose that $Q\subseteq\{0,1\}^k$ is nonempty and majority-closed, and $W\subseteq\Fp^k$ is nonempty and affine with $D_qW\subseteq W$ for every $q\in Q$. Then
\begin{equation}\label{eq:activepair}
 R(Q,W)=\{t\in D_p^k:\chi(t)\in Q,\ \eta(t)\in W\}
\end{equation}
is nonempty and $e_p$-closed, and its activity image is exactly $Q$. Indeed, $(q,D_qw)$ describes a tuple for every $q\in Q,w\in W$. For four such tuples, the output activity belongs to $Q$ and the unmasked output value $-v^1+v^2+v^3$ belongs to $W$. Masking keeps it in $W$. The condition $D_qW\subseteq W$ is thus the compatibility needed to combine activity and value constraints.

Partition the coordinates according to equality of their activity functions $q\mapsto q_i$ on $Q_R$. Let $P$ be the always-active coordinates and $Z$ the always-inactive coordinates. Denote the remaining classes of equal nonconstant activity functions by $C_1,\ldots,C_s$.

\begin{lem}[Independent value factors]\label{lem:activeblocks}
Up to coordinate order, the canonical affine space has the form
\begin{equation}\label{eq:activeblocks}
 W_R=A_P\times\{0\}^{Z}\times\prod_{j=1}^s L_{C_j},
\end{equation}
where $A_P\subseteq\Fp^P$ is a nonempty affine space and each $L_{C_j}\leq\Fp^{C_j}$ is linear. Conversely, every nonempty majority-closed $Q$ and spaces of this form give the canonical pair of an $e_p$-closed relation.
\end{lem}
\begin{proof}
Write $W_R=w_0+L$. By Lemma~\ref{lem:activefibres}, $D_qL\subseteq L$ and $D_qw_0-w_0\in L$. Thus $L$ is invariant under $D_q$ and $I-D_q$ for every $q\in Q_R$.

These commuting diagonal maps generate the projection $E_C$ onto each class $C$ of equal activity functions, including nonempty $P$ and $Z$. For every other class $C'$, choose an activity vector distinguishing $C$ from $C'$. Use $D_q$ if $q$ is one on $C$, and $I-D_q$ otherwise. The product is one on $C$ and zero on every other class. Hence $E_CL\subseteq L$, and $L$ is the direct sum of its class projections.

For a variable class $C$, choose $q$ with $q_C=0$. Projecting $D_qw_0-w_0\in L$ onto $C$ gives $-w_{0,C}\in E_CL$, so the affine offset on $C$ vanishes. An offset may remain on $P$. All vectors in $\eta(R)$, and therefore in $W_R$, are zero on $Z$. This proves~\eqref{eq:activeblocks}.

Conversely, a product of the displayed form is invariant under every allowed activity mask, so~\eqref{eq:activepair} defines an invariant relation. The factor $A_P$ is realized under every activity vector. Each variable class is active under some activity vector and inactive under another. Choosing zero values in the other variable classes realizes every direction in $L_C$. Thus the affine hull of the realized value vectors is exactly the displayed product, proving canonicity.
\end{proof}

For example, when $p=2$ the relation
\[
 \{(\bot,\bot,\bot)\}\ \cup\
 \{(a,b,c)\in\Ftwo^3:a+b+c=0\}
\]
has activity relation $\{000,111\}$ and value space the even-parity plane. Its three coordinates form one variable activity class. The inactive fibre contains one tuple and the active fibre contains four.

We encode $Q_R$ by all its unary and binary projection tables in coordinate order. Majority-closed Boolean relations are determined by these projections, equivalently by 2-CNF formulas, as in the standard small-projection characterization \cite{Berman2010,IMMVW2010}. Encode $W_R$ by a reduced row echelon basis of its direction space and the unique coset representative whose pivot entries are zero. We call the resulting pair the \emph{canonical normal form}. Empty and nullary relations use the flags from Section~\ref{sec:regular}.

\begin{thm}[Canonical normal forms]\label{thm:normalcompiler}
Every $e_p$-closed relation on $k$ coordinates has a canonical normal form using $O(k^2)$ bits. Under the preservation promise, it is computable in polynomial time from an NFA and an explicit arity. The normal form of the relation generated by any explicit tuple set is also computable in polynomial time.
\end{thm}
\begin{proof}
There are $O(k^2)$ Boolean projection tables, each using at most four bits. The affine code uses at most $k^2+k$ field entries. Since $p$ is fixed, the combined bit size is $O(k^2)$. Lemma~\ref{lem:activefibres} proves that it determines the relation.

Let $R=R(M,k)$ be nonempty. For each unary or binary coordinate set and Boolean assignment on it, retain only letters with the specified activity at the pinned positions. Layered reachability computes the corresponding entry of $Q_R$, as in Lemma~\ref{lem:pincompiler}.

To compute $W_R$, start with the value vector of one accepted word and maintain the affine span $H$ of the vectors collected. Gaussian elimination supplies at most $k$ independent equations $\lambda\cdot v=c$ defining $H$. To find an accepted word violating such an equation, augment each layered NFA state by an accumulator $z\in\Fp$. A letter $b$ at position $i$ updates it by
\[
 z\longmapsto z+\lambda_i\eta(b).
\]
At the final layer accept only accumulators different from $c$. If $M$ has $s$ states, this graph has at most $(k+1)sp$ states. A successful test returns a value vector outside $H$. Add it and restart the equation tests. Each addition increases the affine dimension, so at most $k$ additions and $O(k^2)$ tests occur. When all equations hold on all accepted words, $H=W_R$. The procedure computes the exact affine hull for any NFA. Preservation is used to reconstruct $R$ from the two data sets, not to perform the reachability tests.

Now let $F\subseteq D_p^k$ be an arbitrary nonempty set of generators. The activity image of $\clos{F}{e_p}$ is the majority closure of $\chi(F)$. Every Boolean unary or binary relation is majority-closed. For a binary relation, the two sets of at least two input indices supporting the output bits of a majority intersect, so the output pair is one of the three input pairs. Therefore the unary and binary projection tables of the generated activity relation are exactly those of $\chi(F)$.

Comparing the activity columns of $F$ determines $P,Z$, and the classes $C$. Constant columns and equal columns are preserved by generation. Set
\begin{equation}\label{eq:activegenerators}
 A_P=\aff\{\eta(f)_P:f\in F\},\qquad
 L_C=\Span\{\eta(f)_C:f\in F\},
\end{equation}
and form~\eqref{eq:activeblocks}. The resulting relation is closed and contains $F$, hence contains $\clos{F}{e_p}$. Conversely, the factorization of $\clos{F}{e_p}$ given by Lemma~\ref{lem:activeblocks} must contain the spaces in~\eqref{eq:activegenerators}. Its activity relation is the one already computed. This gives the reverse inclusion. Column comparison and Gaussian elimination are polynomial in the explicit input length. Empty inputs and nullary arities are handled by their flags.
\end{proof}

The generator construction does not require its input tuples to witness the forks and small projections of an edge representation. It uses the explicit decomposition particular to $e_p$.

Figure~\ref{fig:normalcompilers} records the two constructions as symbolic algorithms. In each routine, $T$ is the list of unary and binary activity tables, rather than a list of all satisfying activity vectors. The notation $\operatorname{Can}(W)$ denotes the affine origin and reduced row echelon basis used in the canonical code.

\begin{figure}[htbp]
\begin{minipage}{\linewidth}
\small
\begin{tabbing}
\hspace{1.3em}\=\hspace{1.3em}\=\hspace{1.3em}\=\kill
\textsc{Compile-NFA}$(M,k)$\\
\>Find an accepted length-$k$ word $u$. If none exists, return the empty flag.\\
\>Build $T$ by reachability with every unary or binary activity assignment pinned.\\
\>Set $H=\{\eta(u)\}$.\\
\>Repeat:\\
\>\>Compute independent equations defining $H$.\\
\>\>For each equation, search for an accepted word violating it.\\
\>\>If all searches fail, return $(T,\operatorname{Can}(H))$.\\
\>\>For a returned word $w$, set $H=\aff(H\cup\{\eta(w)\})$.\\[1ex]
\textsc{Compile-Generators}$(F,k)$\\
\>If $F$ is empty, return the empty flag.\\
\>Read $T$ directly from the unary and binary projections of $\chi(F)$.\\
\>Compare activity columns to find $P,Z$, and the variable classes $C$.\\
\>Compute $A_P$ and the spaces $L_C$ by~\eqref{eq:activegenerators}.\\
\>Return $(T,\operatorname{Can}(A_P\times\{0\}^{Z}\times\prod_C L_C))$.
\end{tabbing}
\end{minipage}
\caption{Symbolic compilation for the prime-field family. At arity zero, a nonempty output is the flag for $\{()\}$.}
\label{fig:normalcompilers}
\end{figure}

For explicit running-time bounds, put $K=k+1$ and $L=1+|Q|+|\Delta|$ for the input NFA. We count elementary field and indexed-graph operations, with $p$ fixed. Using predecessor pointers for path reconstruction, each layered reachability query costs $O(KL)$. The activity tables require $O(K^2)$ queries, and the affine-hull loop requires $O(K^2)$ further queries. Recomputing equations by classical Gaussian elimination at each of its at most $K$ stages costs $O(K^4)$ in total. Thus \textsc{Compile-NFA} takes
\begin{equation}\label{eq:normalnfatime}
 O(K^3L+K^4)
\end{equation}
operations and $O(KL+K^2)$ working storage. For $m=|F|$, scanning the tables and comparing activity columns costs $O(mK^2)$. Row reduction on the factor matrices costs $O(mK^2+K^3)$, since the sum of the squared factor sizes is at most $k^2$. Hence \textsc{Compile-Generators} takes
\begin{equation}\label{eq:normalgentime}
 O(mK^2+K^3)
\end{equation}
operations. These are conservative bounds, and both algorithms return only the $O(k^2)$-bit code.

\begin{cor}[Normal-form calculus]\label{cor:normalcalculus}
Products, intersections, coordinate identifications, and ordered projections are computable in polynomial time on canonical normal forms. Automatic instances preserved by $e_p$ have canonical $O(k^2)$-bit representations of their solution relations on any ordered boundary of length $k$, computable in polynomial time in the instance size and $k$.
\end{cor}
\begin{proof}
Products use the pairs $(Q_1\times Q_2,W_1\times W_2)$. On the same coordinates, an intersection is represented by
\[
 Q=Q_1\cap Q_2,\qquad W=W_1\cap W_2.
\]
Conjoin the 2-CNF formulas and the affine equation systems. If either system is inconsistent, the relation is empty. Otherwise, every $q\in Q$ preserves both $W_i$ under masking, and therefore preserves their intersection. The converse following Lemma~\ref{lem:activefibres} shows that every $q\in Q$ is realized. Hence the two separate feasibility tests suffice.

The intersection may retain redundant value directions on the coordinates $Z$ now forced inactive. Its canonical value space is
\begin{equation}\label{eq:activezero}
 W^0=W\cap\{v:v_Z=0\}.
\end{equation}
Indeed, the diagonal-projector proof of Lemma~\ref{lem:activeblocks} applies to the mask-stable affine space $W$. It gives an affine factor on the always-active coordinates and linear factors on the other activity classes, including $Z$. Realized tuples discard exactly the factor on $Z$, and their affine hull realizes all remaining factors. Compute $Z$ and all feasible unary and binary assignments by 2-SAT, and canonicalize~\eqref{eq:activezero} by Gaussian elimination.

Identifying coordinates $i,j$ adds $q_i=q_j$ to the activity formula and $v_i=v_j$ to the affine system. Every remaining activity mask preserves this field equality, so the preceding intersection argument applies.

For an ordered coordinate list $\mathbf b$, possibly with repetitions, put
\[
 Q'=\pr_{\mathbf b}Q,\qquad W'=\pr_{\mathbf b}W.
\]
Every projected tuple lies in $R(Q',W')$. Conversely, extend the activity of a tuple of $R(Q',W')$ to $q\in Q$ and its value to $w\in W$. The two extensions may initially be incompatible, but $D_qw\in W$. Masking preserves all selected values, since the value at every selected inactive position is zero. Thus $(q,D_qw)$ describes a full tuple realizing the proposed projection. Projection preserves canonicity because coordinate projection commutes with the affine hull and with $\eta$. The tables of $Q'$ are obtained by 2-SAT tests, and projecting and row-reducing an affine origin and basis gives $W'$.

Compile each automatic constraint by Theorem~\ref{thm:normalcompiler}, introduce its scope coordinates, identify them with the instance variables, and intersect. Free variables use the full activity relation and full field space. All scopes may repeat variables. The complete solution relation and its boundary projections are then computed by the preceding operations. Given a satisfying activity vector $q$ and an affine solution $w$, the legal pair $(q,D_qw)$ supplies a witness assignment. Empty and nullary cases use the designated flags.
\end{proof}

Figure~\ref{fig:normaljoin} gives one routine combining the product, equality, and projection steps. Intersection is obtained by equating corresponding coordinates of two equal-arity factors. Every forbidden unary or binary table entry contributes one clause to the 2-CNF formula. Thus the formula has $O(K^2)$ clauses when $K$ bounds the sum of the input and output arities plus one. Computing all feasible unary and binary assignments requires $O(K^2)$ 2-SAT tests, each taking $O(K^2)$ time by computing the strongly connected components of the implication graph. A spanning forest reduces the identification list to $O(K)$ equalities. Affine intersection, restriction to zero coordinates, and projection then cost $O(K^3)$ by classical elimination. Consequently, each product, intersection, coordinate identification, or ordered projection can be computed in $O(K^4)$ elementary operations. This also bounds the routine in Figure~\ref{fig:normaljoin} for $O(K^2)$ listed identifications.

\begin{figure}[htbp]
\begin{minipage}{\linewidth}
\small
\begin{tabbing}
\hspace{1.3em}\=\hspace{1.3em}\=\kill
\textsc{Join-Project}$((T_1,W_1),(T_2,W_2),E,\mathbf b)$\\
\>If either factor is empty, return the empty flag.\\
\>Form a 2-CNF formula $\Phi$ for $Q_1\times Q_2$.\\
\>Form an affine equation system $\mathcal A$ for $W_1\times W_2$.\\
\>Replace $E$ by a spanning forest of its undirected equality graph.\\
\>For every $(i,j)\in E$, add $q_i=q_j$ to $\Phi$ and $v_i=v_j$ to $\mathcal A$.\\
\>If $\Phi$ or $\mathcal A$ is inconsistent, return the empty flag.\\
\>Use unary 2-SAT tests to find $Z=\{i:\Phi\text{ forces }q_i=0\}$.\\
\>Let $W$ solve $\mathcal A$ together with $v_i=0$ for $i\in Z$.\\
\>Build $T'$ by pinning each unary or binary assignment on the positions in $\mathbf b$.\\
\>Return $(T',\operatorname{Can}(\pr_{\mathbf b}W))$.
\end{tabbing}
\end{minipage}
\caption{Normal-form composition. The equality list $E$ uses the concatenated input coordinates. Conflicting pins at repeated positions fail their 2-SAT test. A nonempty nullary result uses its designated flag.}
\label{fig:normaljoin}
\end{figure}

For an automatic instance of size $N$ from~\eqref{eq:size}, summing~\eqref{eq:normalnfatime} over its constraints costs $O((N+1)^4)$. The insertion algorithm makes $O(N+1)$ normal-form updates, followed by one projection to the requested $k$-coordinate boundary. The preceding bounds therefore give the conservative total
\begin{equation}\label{eq:normalinstancetime}
 O((N+k+1)^5).
\end{equation}
The exponent is independent of the fixed prime $p$. Field-operation constants may depend on $p$, and ordinary binary encodings of indices preserve polynomial time.

A worked example exhibits the interaction between the activity formula and the affine factors. Take $p=3$ and coordinates $(u_1,u_2,x_1,x_2,y_1,y_2)$. Let
\begin{equation}\label{eq:workedpair}
 \begin{aligned}
 Q&=\{(1,1,a,a,b,b):a,b\in\{0,1\},\ a\leq b\},\\
 W&=\{v\in\mathbb F_3^6:v_1+v_2=1,\ v_3=v_4,\ v_6=2v_5\}.
 \end{aligned}
\end{equation}
Here $P=\{1,2\}$, $C_1=\{3,4\}$, and $C_2=\{5,6\}$. The always-active factor is the affine line $u_1+u_2=1$. The variable factors are $\Span\{(1,1)\}$ and $\Span\{(1,2)\}$, and the activity formula links them by $a\Rightarrow b$. Lemma~\ref{lem:activeblocks} shows that $R=R(Q,W)$ is invariant. Its three fibres are
\[
\begin{array}{cclc}
 (a,b)& &\text{tuples, with }u,x,y\in\mathbb F_3&\text{size}\\
 (0,0)&& (u,1-u,\bot,\bot,\bot,\bot)&3\\
 (0,1)&& (u,1-u,\bot,\bot,y,2y)&9\\
 (1,1)&& (u,1-u,x,x,y,2y)&27.
\end{array}
\]
Thus $R$ has $39$ tuples, with independent field parameters in each active factor.

Project first to the ordered list $\mathbf b=(1,3,5,6,3)$. The repeated third coordinate of the original relation appears in the second and fifth output positions. The canonical pair becomes
\[
 \begin{aligned}
 Q_{\mathbf b}&=\{(1,a,b,b,a):a\leq b\},\\
 W_{\mathbf b}&=\{(u,x,y,2y,x):u,x,y\in\mathbb F_3\}.
 \end{aligned}
\]
The activity implication remains, while projection makes the always-active value $u$ free. The projected relation still has $39$ tuples, since the omitted coordinates are determined by the retained ones.

For an intersection on the original coordinates, require $y_1=\bot$. Its activity constraint is $b=0$, so $a\Rightarrow b$ also forces $a=0$. Intersecting the affine systems initially gives
\[
 \widetilde W=\{(u,1-u,x,x,0,0):u,x\in\mathbb F_3\}.
\]
This still contains the $x$ direction. The newly inactive set is $Z=\{3,4,5,6\}$, and~\eqref{eq:activezero} removes that direction to give
\[
 W^0=\{(u,1-u,0,0,0,0):u\in\mathbb F_3\}.
\]
The resulting relation consists of the three tuples $(u,1-u,\bot,\bot,\bot,\bot)$. This calculation makes explicit why the affine intersection must be canonicalized after solving the activity constraints.

A primitive positive, or pp-, definition is an existentially quantified conjunction of relation atoms and equalities. To obtain short definitions, let $\Gamma_p$ contain every lifted unary and binary Boolean relation $\chi^{-1}(B)$, all active singletons $\{c\}$ with $c\in\Fp$, the unary relation $Z_*=\{\bot,0\}$, and
\begin{equation}\label{eq:switchedaddition}
 S_p=\{(\bot,\bot,\bot)\}\cup\{(a,b,c)\in\Fp^3:a+b=c\}.
\end{equation}
These are finitely many relations of arity at most three, all preserved by $e_p$.

\begin{prop}[Constructive short definitions]\label{prop:normaldefinitions}
The pp-definable relations over $\Gamma_p$ are exactly the $e_p$-closed relations. A pp-definition with $O(k^2)$ atoms and auxiliary variables is computable in polynomial time from the normal form of a $k$-ary relation. The quadratic normal-form bit bound is optimal in the worst case.
\end{prop}
\begin{proof}
Use the lifted unary and binary tables to define the activity relation. Gaussian elimination gives at most $|P|$ affine equations on $P$ and $|C|$ homogeneous equations on each variable class $C$ in~\eqref{eq:activeblocks}. Realize each equation by a chain of $S_p$ constraints adding its nonzero terms. Coefficients are implemented by repetitions modulo $p$, at constant cost for fixed $p$. For an affine equation on $P$, pin the initial accumulator to the active zero and the final accumulator to the required constant. For a homogeneous equation on $C$, put both endpoint accumulators in $Z_*$.

Each addition constraint equates its three activity bits. Thus a homogeneous chain is entirely inactive or entirely active. The inactive chain is satisfied, while an active chain has both endpoints equal to the field zero and enforces the homogeneous equation. The activity formula already ties every coordinate of $C$, including coordinates absent from an individual equation. Tautological zero equations are omitted, and an inconsistent equation is represented using the empty unary relation, which belongs to $\Gamma_p$.

There are at most $|P|^2+\sum_C|C|^2\leq k^2$ equation entries. This proves the construction bound. Every invariant relation has the resulting definition, and every pp-definition over $\Gamma_p$ is preserved by $e_p$, proving the characterization.

For sharpness, let $k=2m$. The active relations
\[
 R_A=\{(x,Ax):x\in\Fp^m\}\subseteq D_p^{2m},\qquad A\in\Fp^{m\times m},
\]
are invariant and are distinct for distinct matrices. An injective binary encoding must distinguish $p^{m^2}$ relations. If its maximum length is $b$, then $2^{b+1}>p^{m^2}$, giving $b\geq m^2\log_2p-1$. For odd arities, append an active coordinate fixed to zero. This is $\Omega(k^2)$ for fixed $p$.
\end{proof}

The definition bound counts atoms and auxiliary variables. An ordinary binary encoding of the resulting formula uses $O(k^2\log(k+2))$ bits for the variable indices, whereas the canonical table-and-matrix code uses $O(k^2)$ bits. The construction from generators in Theorem~\ref{thm:normalcompiler}, followed by Proposition~\ref{prop:normaldefinitions}, gives the short-definition algorithm for this family.

\section{Graphoid automata and effective boundaries}\label{sec:graphs}

We use the canonical unitary relational semantics of graphoid automata from \cite{BK2004,BK2008}. Let $\Sigma$ be a finite doubly ranked alphabet. A label $\sigma\in\Sigma$ has an input rank $p_\sigma$ and an output rank $q_\sigma$. A finite $\Sigma$-graph $H$ has a finite vertex set, labelled hyperedges with ordered input and output incidence lists of the appropriate ranks, and ordered external input and output lists $\mathbf i_H$ and $\mathbf o_H$. Incidence and boundary lists may repeat vertices. Isolated vertices and empty boundaries are allowed.

A graphoid automaton is $A=(\Sigma,D,\delta,I,T)$, where $D$ is a finite state set,
\[
 \delta_\sigma\subseteq D^{p_\sigma}\times D^{q_\sigma}
\]
is the transition relation for $\sigma$, and $I,T\subseteq D^*$ are regular languages. A run on $H$ is a map $h:V(H)\to D$ satisfying the transition relation at every edge. It is accepting when
\[
 h(\mathbf i_H)\in I,\qquad h(\mathbf o_H)\in T.
\]
This vertex-assignment description is the existential relational interpretation of graph gluing. In particular, serial composition identifies matching output and input vertices, while parallel composition takes disjoint unions and concatenates the boundary lists.

Let $\mathbf b_H$ be the concatenation of the input and output lists. The \emph{boundary relation} of $H$ is
\begin{equation}\label{eq:graphboundary}
 R_H=\{h(\mathbf b_H):h\text{ is a run on }H\}.
\end{equation}
This relation records extendibility before applying $I$ and $T$. One can also use any other explicit ordered list of vertices as the boundary. After imposing the regular acceptance constraints, the corresponding projection records the boundary tuples of accepting runs.

\begin{thm}[Recognition and boundary representations]\label{thm:graphoid}
Let $A$ be a fixed graphoid automaton. Suppose a common edge operation $e$ on $D$ preserves every transition relation and every length slice of $I$ and $T$. Recognition of finite input graphs is polynomial-time, with no restriction on their width. Compact representations of their boundary relations and of the boundary tuples of accepting runs are computable in polynomial time.

For an $\ell$-edge operation, a boundary of length $k\geq1$ has a representation with $O(k^{\ell-1})$ tuples and $O(k^\ell)$ bits. If the common operation is \Mal, at most $2kd^2$ tuples suffice, giving $O(k^2)$ bits for fixed $D$.
\end{thm}
\begin{proof}
Use one variable per vertex. Each hyperedge supplies a constraint with its concatenated incidence list as scope and its transition relation as the allowed relation. Since $A$ is fixed, a trie NFA for each transition relation has fixed size. The two ordered external lists supply the constraints described by NFAs for $I$ and $T$.

The resulting automatic instance has size polynomial in the explicit graph description and has exactly the accepting runs as solutions. Apply Theorem~\ref{thm:edge}, or Theorem~\ref{thm:maltsev}. For $R_H$ itself, omit the two acceptance constraints and project the local solution relation. All operations allow repeated coordinates. The number of bits follows by storing the bounded number of length-$k$ tuples over the fixed alphabet $D$.
\end{proof}

The hypothesis involves the local transitions and the boundary languages together. Regularity of $I,T$ alone supplies no common polymorphism. In the fixed-automaton setting, Proposition~\ref{prop:dfa} provides a finite check for a proposed operation when deterministic boundary automata are supplied.

\begin{cor}[Effective relational composition]\label{cor:composition}
For graphs whose transition relations share a fixed edge operation, serial composition, parallel composition, and boundary identification can be implemented on compact boundary representations in polynomial time. Inclusion and equality of the relations of two specified finite graphs are also decidable in polynomial time, with a boundary counterexample whenever inclusion fails.
\end{cor}
\begin{proof}
For parallel composition, take a product and permute the coordinates into the input-output order. For serial composition, take a product, identify every matched output-input pair, and project away the glued interface. Boundary identifications are equality updates. Lemma~\ref{lem:edgecalculus} performs these operations. The semantics of gluing proves that the resulting relations are the boundary relations of the composed graphs. Comparison follows from Corollary~\ref{cor:comparison}.
\end{proof}

The comparison statement concerns two finite graph relations of the same boundary type. It is not a statement about equivalence of the languages recognized by two graphoid automata on all input graphs. The summaries may be evaluated along a supplied composition expression, but the polynomial-time recognition theorem itself does not require such an expression or a width bound.

For the family from Section~\ref{sec:normalforms}, the boundary relation has a canonical code and a small graph realization over a fixed transition alphabet.

\begin{cor}[Canonical quadratic boundaries]\label{cor:normalboundaries}
For a fixed graphoid automaton on $D_p$ whose transitions and boundary slices are preserved by $e_p$, canonical $O(k^2)$-bit boundary representations are computable in polynomial time. Their composition is computable by 2-SAT and Gaussian elimination. Over the fixed transition alphabet $\Gamma_p$, an equivalent graph with $O(k^2)$ vertices and edges is constructible from the code of a $k$-coordinate boundary relation. The quadratic bit bound is optimal for this alphabet.
\end{cor}
\begin{proof}
Apply Corollary~\ref{cor:normalcalculus} to the vertex-assignment instance in the proof of Theorem~\ref{thm:graphoid}. The same argument includes the regular acceptance constraints when required. To realize a normal form as a graph, take one vertex for each free or existential variable in Proposition~\ref{prop:normaldefinitions} and one edge for each relation atom. Give an arity-$r$ relation the rank $(r,0)$, and use the free variables in order as the boundary, with the prescribed input-output split. Repeated variables become repeated incidences, and all intermediate variables are internal vertices. The pp-definition and the graph then have the same boundary relation. The relations $R_A$ used in Proposition~\ref{prop:normaldefinitions} are realizable over this single fixed alphabet, so its lower bound applies.
\end{proof}

The graph replacement in Corollary~\ref{cor:normalboundaries} uses $\Gamma_p$. The original transition alphabet need not contain all these relations. The size bound concerns vertices and edges, while the canonical boundary code has the stated quadratic bit size.

The quadratic \Mal\ bound also cannot be improved for exact boundary encodings across the whole class. The argument uses a fixed Boolean automaton and a family of graphs of polynomial size.

\begin{prop}[Quadratic lower bound]\label{prop:lowerbound}
There is a fixed two-state graphoid automaton with a common \Mal\ polymorphism such that any injective binary encoding of all its $k$-coordinate boundary relations requires $\Omega(k^2)$ bits in the worst case.
\end{prop}
\begin{proof}
Take $D=\Ftwo$, with $p(x,y,z)=x+y+z$. Use the ternary transition
\[
 X=\{(a,b,c):a+b=c\}
\]
and the unary transition $Z=\{0\}$. Both are preserved by $p$. For an integer $m$, every binary $m\times m$ matrix $A$ defines the relation
\[
 R_A=\{(x,y)\in\Ftwo^m\times\Ftwo^m:y=Ax\}.
\]
It is realized by a graph with boundary $(x_1,\ldots,x_m,y_1,\ldots,y_m)$. For each row of $A$, initialize an accumulator with a $Z$-edge and use $X$-edges to add the variables in that row, with the final accumulator identified with the corresponding $y$-vertex. An empty row is enforced by a $Z$-edge on $y$. This uses $O(m^2)$ vertices and edges. Labels may be taken of ranks $(3,0)$ and $(1,0)$, so all these are graphs over one fixed alphabet.

Distinct matrices give distinct relations, since their values differ on some Boolean vector. Hence there are at least $2^{m^2}$ different boundary relations of length $2m$. An injective encoding with maximum length $b$ uses fewer than $2^{b+1}$ binary strings, so $b\geq m^2-1$. For odd boundary lengths, add one coordinate fixed to zero. Thus the lower bound is $\Omega(k^2)$ for arbitrary $k$. The regular acceptance languages may be taken to be $D^*$, so they impose no additional restriction.
\end{proof}

This is an information bound for exact relations. It does not impose a lower bound on the time needed for a single recognition query or on encodings of a more restricted family of graph relations.

\section{A three-state separation from counting}\label{sec:counting}

For a finite group $G$, every nonempty relation preserved by the heap operation $p(x,y,z)=xy^{-1}z$ is a coset of a subgroup of a power of $G$. Indeed, translating such a relation by one of its elements gives a set containing the identity and closed under multiplication and inverse. Projections of group cosets have uniform nonempty fiber sizes. General \Mal\ relations need not have this property, as the following example makes explicit.

Let $D=\{0,1,2\}$ and define $\chi:D\to\Ftwo$ by $\chi(0)=0$ and $\chi(1)=\chi(2)=1$. Set
\begin{equation}\label{eq:inflated}
 p(a,b,c)=
 \begin{cases}
 c,&a=b,\\
 a,&a\ne b\text{ and }b=c,\\
 \chi(a)+\chi(b)+\chi(c),&\text{otherwise},
 \end{cases}
\end{equation}
where the last value is evaluated in $\Ftwo$ and represented by $0$ or $1$ in $D$.

\begin{lem}\label{lem:inflated}
The operation~\eqref{eq:inflated} is \Mal\ and satisfies
\[
 \chi(p(a,b,c))=\chi(a)+\chi(b)+\chi(c).
\]
It preserves the five-element relation
\[
 E=\{(a,b)\in D^2:\chi(a)=\chi(b)\},
\]
which is not a coset in the square of any group structure on $D$ and has nonuniform projection fibers.
\end{lem}
\begin{proof}
The first two branches give the \Mal\ identities. In each branch, applying $\chi$ gives the displayed sum, with repeated terms cancelling. Preservation of $E$ follows by applying this identity in both coordinates.

The relation is $\{(0,0),(1,1),(1,2),(2,1),(2,2)\}$, so it has size five. A coset in the square of a group of order three must have size dividing nine. Hence $E$ cannot be such a coset. Its fibers over the first coordinate have sizes $1,2,2$.
\end{proof}

The example also supplies a fixed counting obstruction. We count vertex assignments, not accepting paths of an NFA or different algebraic expressions for the same run.

\begin{thm}[Recognition versus counting]\label{thm:counting}
There is a fixed three-state graphoid automaton whose transitions have a common \Mal\ polymorphism, for which recognition is polynomial-time but counting accepting vertex assignments is $\sharpP$-complete under polynomial-time Turing reductions. Hardness holds on graphs with empty boundaries and only unary and ternary transition relations.
\end{thm}
\begin{proof}
Use the operation in~\eqref{eq:inflated}, the unary relation $U=\{0,1\}$, and
\[
 C=\{(a,b,c)\in D^3:\chi(c)=\chi(a)+\chi(b)\}.
\]
Include the singleton transitions $\{0\}$ and $\{1\}$ as well. The operation restricts to Boolean minority on $U$, so $U$ is preserved. Singleton relations are preserved by idempotence. Applying the identity of Lemma~\ref{lem:inflated} shows that $C$ is preserved. Take $I=T=\{\varepsilon\}$. Recognition is polynomial by Theorem~\ref{thm:graphoid}.

For variables $x,y$ restricted to $U$, the number of extensions to a fresh unrestricted variable $z$ satisfying $C(x,y,z)$ is the Boolean function $f$ with matrix
\begin{equation}\label{eq:weight}
 (f(a,b))_{a,b\in\{0,1\}}=
 \begin{pmatrix}1&2\\2&1\end{pmatrix}.
\end{equation}
For a finite list of ordered pairs $(x_j,y_j)$ of Boolean variables, consider the weighted counting problem
\[
 Z_f=\sum_{h:V\to\{0,1\}}\prod_j f(h(x_j),h(y_j)).
\]
Replace every factor by a constraint $C(x_j,y_j,z_j)$ with its own fresh variable $z_j$, and impose $U$ on each original variable. For any assignment of the original variables, the fresh variables have independent extension choices. Therefore the number of satisfying assignments of the resulting unweighted instance is exactly $Z_f$.

The weighted Boolean dichotomy \cite[Theorem~4]{DGJ2009} makes evaluation of $Z_f$ $\sharpP$-hard. To check its hypotheses directly, $f$ has full support and unequal positive values, so it is not pure affine. A full-support product-type binary function must factor into unary weights and consequently has matrix rank one. The determinant of~\eqref{eq:weight} is $-3$, so $f$ is not of product type either.

Interpret each unary or ternary constraint as an edge of ranks $(1,0)$ or $(3,0)$ and use empty boundaries. This gives the required graph instance over a fixed alphabet. Counting is in $\sharpP$, since a vertex assignment can be verified in polynomial time using the fixed finite relations. The reduction preserves the integer count, and the hardness from the weighted dichotomy uses polynomial-time Turing reductions.
\end{proof}

The failure of general \Mal\ preservation to guarantee tractable counting is already part of the counting-CSP landscape. Theorem~\ref{thm:counting} gives an explicit realization in the graphoid model. Together with Lemma~\ref{lem:inflated}, it explains why a compact support relation cannot in general be supplemented with one uniform multiplicity to recover all extension counts.

\section{Concluding remarks}

The regular presentation provides exactly the two kinds of witnesses required by compact representation algorithms: exchanges after a common prefix and assignments on bounded sets of coordinates. Once these witnesses have been extracted, arbitrary scopes are assembled through products and binary equalities. This separates the automaton compilation from the established algebraic updates and gives a polynomial-time algorithm whose complexity is measured in the succinct input size.

The prime-field family supplies a second representation mechanism. Its masking term reconstructs each activity fibre from one affine value hull, and the common activity classes determine independent value factors. This produces optimal quadratic-bit codes beyond the \Mal\ and near-unanimity cases, with direct constructions from NFAs and arbitrary generators. The result concerns this explicit family. Determining whether equally small effective encodings exist for arbitrary 3-edge algebras remains open here.

For graphoid automata, the resulting boundary relations support exact finite-graph comparison and composition without a width assumption. The quadratic bounds are optimal as encoding bounds, but leave room for faster updates and smaller representations under further restrictions on the transition language. The fixed basis $\Gamma_p$ also permits a boundary relation to be replaced constructively by a graph of quadratic combinatorial size.

The counting example suggests a further distinction worth studying: which common edge polymorphisms, together with which transition languages, permit efficient representations of extension multiplicities as well as supports? Another question concerns succinct descriptions beyond NFAs. The proof applies whenever the description admits polynomial-time extraction of all fork witnesses and all bounded projection witnesses. Identifying useful representation formalisms with these two properties would extend the compilation theorem while keeping its algebraic part unchanged.

Reference implementations and executable verification are archived on Zenodo as version~1.0.0 of the reproducibility package~\cite{Kalampakas2026Software}. The package implements NFA witness extraction, the direct \Mal\ calculus, the general edge-representation compiler, and the symbolic NFA and generator compilers of Section~\ref{sec:normalforms}. Independent finite-operation closure checks test the generator construction on small relations. Further checks cover the worked example, intersections, ordered projections, and arity-$64$ inputs whose represented relations are not enumerated. Execution instructions, dependency specifications, reference outputs, checksums, and a combined verification runner are included. The documentation specifies the scope of each routine and distinguishes symbolic compilation from enumerating verification oracles.

\bibliographystyle{alphaurl}
\bibliography{mybib}
\end{document}